\documentclass[submission,copyright,creativecommons]{eptcs}
\providecommand{\event}{TARK 2015} 
\usepackage{breakurl}             

\usepackage{graphicx}
\usepackage{amsfonts}
\usepackage{amsthm}

\newcommand{\commentout}[1]{}

\newtheorem{theorem}{Theorem}
\newtheorem{definition}{Definition}
\newtheorem{axiom}{Axiom}

\newtheorem{reminder}{Reminder}

\title{An Axiomatic Approach to Routing}
\author{Omer Lev
\institute{Hebrew University and\\Microsoft Research\\ Israel}
\email{omerl@cs.huji.ac.il}
\and
Moshe Tennenholtz
\institute{Technion\\ \\
Haifa, Israel}
\email{moshet@ie.technion.ac.il}
\and
Aviv Zohar
\institute{Hebrew University and\\Microsoft Research\\ Israel}
\email{avivz@cs.huji.ac.il}
}
\def\titlerunning{An Axiomatic Approach to Routing}
\def\authorrunning{O. Lev, M. Tennenholtz \& A. Zohar}
\begin{document}
\maketitle

\begin{abstract}
Information delivery in a network of agents is a key issue for large, complex systems that need to do so in a predictable, efficient manner. The delivery of information in such multi-agent systems is typically implemented through routing protocols that determine how information flows through the network. Different routing protocols exist each with its own benefits, but it is generally unclear which properties can be successfully combined within a given algorithm. We approach this problem from the axiomatic point of view, i.e., we try to establish what are the properties we would seek to see in such a system, and examine the different properties which uniquely define common routing algorithms used today.

We examine several desirable properties, such as robustness, which ensures adding nodes and edges does not change the routing in a radical, unpredictable ways; and properties that depend on the operating environment, such as an ``economic model'', where nodes choose their paths based on the cost they are charged to pass information to the next node. We proceed to fully characterize minimal spanning tree, shortest path, and weakest link routing algorithms, showing a tight set of axioms for each.
\end{abstract}

\section {Introduction}
 
The proper way to distribute power, disseminate information, or establish hierarchies in organizations is an issue encountered whenever there is a large enough network of agents that needs to interact in an orderly manner. For example, when trying to establish efficient lines of communications between agents which all need to reach a central hub, there are various properties we may desire in our system. We might want the system to be able to handle small changes in connections without causing disruptions throughout the network; we may want it to be flexible when we change its parameters so that various routing options are possible, and more. Indeed, the search for the right communication structure has played a role in early work on the foundations of the area of multi-agent systems \cite{Fox81,Mal86,DLC87}, based on classical work in organization theory \cite{Gal73,MS58}.

More concretely, examining networking, one of the most important aspects of the design of a communication network is the way it routes information through its physical links. Routing protocols, such as those used in packet switching networks, circuit switching, or ad-hoc networks are designed with many goals in mind. They must adapt to changing network conditions, withstand failures, and operate in a distributed fashion while constructing a ``good'' routing scheme. Nodes in the network are, in fact, autonomous agents that can control the flow of information through them and can choose to forward it according to their own considerations. 
Agents may be controlled by different economic entities (such as in the internet, where different internet service providers control some of the routers), and may route according to complex preferences that are derived from economic relations~\cite{GR01,LSZ08}. Even in the cooperative local-network setting where all routers are controlled by a single network operator, different considerations such as bandwidth utilization, latency, and the risks of link failures come into play.

The multitude of previous treatments of the problem suggest a myriad of routing protocols, each with their own benefits and shortcomings. In contrast, this work examines the routing problem through the lens of the \emph{axiomatic approach}, which seeks to formulate different elementary properties that are desirable in this context. One approach to an axiomatic treatment, which we take in this work, is that of characterization: a set of elementary properties is shown to uniquely determine some routing algorithm, and hence the routing outcome on any specific graph. From the designer's perspective, such a result implies a great deal -- any additional property that is not already achieved by the protocol cannot be added to it without giving up on another basic property. The approach thus provably bounds the design space of algorithms and makes explicit the choices made when selecting one over the other.

As we are not aware of any previous axiomatic treatment of routing, we focus our attention on a domain that most closely resembles the internet as it is built today, and focus our efforts within this domain on what one may consider classic, or natural routing schemes. In particular, we assume that routing choices are independent of the congestion on links (such is the case in the internet, where routing protocols such as BGP first establish paths, and congestion control protocols such as the one embedded into TCP manage the load on each flow's path and ensures that rates are throttled to match the bottleneck of the flow). Furthermore, as with internet routing where routers decide on the next hop of each packet using a routing table that maps its destination to the next hop, routing choices made to different destinations are done independently. Finally, packets addressed to the same destination are not split between different paths, and are routed in the same manner regardless of their source. These choices, which greatly restrict the power of any routing algorithm may seem arbitrary, but are in fact derived from real-world design considerations. For example, the need to quickly forward packets towards their destination at each router mandated that most routing be done in specialized hardware. No complex computation is performed (only a lookup into a routing table) and no deep inspection of the packet is performed. Keeping routing simple has made it fast and robust.

More advanced routing schemes that have been proposed in the literature may split traffic, allow routing choices to depend on the source of the packet or its previous hops, or may even change the routes in response to link congestion. These are notoriously difficult to coordinate and to implement. We leave treatment of these more advanced schemes to future work.

Our set of axioms or ``desirable properties'' are also motivated by similar considerations. For example, one of the fundamental features we desire in our algorithms is one of {\bf robustness}, which is the ability of a system to endure changes in the network without creating disruption in parts of the network that have not undergone changes.

A different feature, which might be desirable only in certain cases, is {\bf ``first hop''}, which is particularly relevant for diffuse networks with independent nodes. It means, broadly, that network nodes care only about their immediate surroundings, or the ``next step'' in the network data transfer. Such a property might be relevant when nodes pursue an ``economic model'', paying for transferring information, and hence only caring about the cost they need to pay to move their information to the next node, and following that, they have no preference on the route the information should pass en route to its destination. Other properties, desirable only in some cases include an indifference between two parallel paths, as long as they change their weights by the same amount concurrently.

Ultimately, after devising our axioms we successfully fully characterized 3 natural routing algorithms:
\begin{itemize}
\item {\bf Minimum spanning tree:} A tree with the smallest overall weight is a result, among others, of the ``first hop'' axiom (the ``economic model'').
\item {\bf Shortest path:} A tree where each node has the shortest possible path to its destination is a result, among other axioms, of viewing as immaterial to the routing decision any parallel paths which change their weight by the same amount.
\item {\bf Weakest link tree:} A tree where each node takes the path with the maximal ``lightest'' weight available to it. This results from considering higher edge-weights as beneficial (e.g., representing bandwidth which one wishes to increase in contrast to delay that one wishes to decrease), and from considering designers that choose between parallel paths in a slightly different manner.
\end{itemize}

We proceed to review relevant previous research and then continue to define our model and expand on the axioms, which are motivated with a brief explanation and presented formally. Following that we show (and prove) our characterization of the minimal spanning tree, the shortest path tree, and the weakest link tree.
\section{Related Work}

In the past decade, as routers became more flexible, research on routing (particularly inter-domain) and its techniques has been rekindled and extended beyond the technical issues dealt with in the past. The harbinger for much of this research was \cite{GS05}, which was further expanded by several researchers (see updating report here: \url{http://www.cl.cam.ac.uk/~tgg22/metarouting/} ). However, this line of research, while introducing many interesting mathematical and theoretical concepts to the field of routing, has refrained from phrasing its models as requirements by users, to be filled by various routing algorithms.

The axiomatic approach, which does approach problems with this outlook, has been first introduced in CS contexts as extensions to the classical theory of choice \cite{Arr51}, and has been applied to ranking systems \cite{AT05b,AT10} and trust systems \cite{ABCFFKMT08}, as well as to other multi-agent setups such as multi-level marketing \cite{EKTZ11}.

In relation to networking, usage of the axiomatic approach has generally been concentrated in two main areas: applying to general graph theory (e.g., \cite{VP88}) or in more technical approaches to networks: papers such as \cite{KNGLYE04} which deal with particular wireless models and implementations, and, somewhat closer to our line of work, \cite{KKPB07}, whose basic axioms are basic enough to be covered through our models, while the routing related axioms involve various assumptions on how routers work (tables, etc.), which we refrain from approaching in our more abstract considerations.

Further work connecting networking and the axiomatic approach has focused on particular instances of problems:  \cite{HSE95} try to use the axiomatic approach to extract the costs of multicast routing and decide who is to pay them. Trust networks and social networks (e.g., recommendation systems) have been analysed many times using the axiomatic approach to understand their desirable features and better understand desirable algorithms in these cases \cite{SYHL06,ABCFFKMT08}. However, none of these papers deal with the basic routing mechanism by which messages and information arrive at each node.
\section{Setup}

Before introducing our axioms, we begin by setting up our routing model. It is, naturally, only a simplification of routing as it is done in large, complex networks such as the internet, but we believe it is robust enough to display many networking characteristics.

Our world will be a weighted graph $G(V,E,W)$ 
and a destination $d$, where $V$ is a set of nodes, $E$ is a set of edges, and $W$ is a function assigning weights to edges, and $d \in V$. A routing solution is a tree $T$ over that graph, as defined below (we do not concern ourselves with non-tree routing, as passing through the same node several times does not serve any purpose).

\begin{definition}
A routing function $f_{d}:\mathcal{G}\rightarrow \mathcal{T}$ is a function from connected weighted graph $G(V,E,W)\in \mathcal{G}$ in which $d\in V$, to a tree $T(V,E,W)\in \mathcal{T}$ such that $T\subseteq G$.
\end{definition}

We can look at the graph as one with directed edges if we consider each edge's direction to be the one pointing at the vertex from which there is a path to $d$ (without going through the same edge again).

We discuss 3 different routing options:
\begin{itemize}
\item \emph{Minumum spanning tree (MST)}: a tree connecting all nodes in the graph with the minimal weight, i.e., for every tree $T'\subseteq G$ that encompasses all of $G$'s nodes, $\sum_{e\in f_{d}(G)}W(e)\leq \sum_{e\in T'}W(e)$.

\item \emph{Shortest path}: each node is connected to $d$ using a shortest length path in the graph. For every node $v\in V$, let $(e_{1},\ldots,e_{s})$ be a path without cycles from $v$ to $d$ such that $e_{i}\in T$, and let $(e'_{1},\ldots, e'_{k})$ a different path from $v$ to $d$, then $\sum_{i=1}^{s}W(e_{i})\leq \sum_{j=1}^{k}W(e'_{j})$.

\item \emph{Weakest link}: looking at each potential path from each node to $d$, we give each path the value of its smallest valued edge. The routing tree will contain, for each node the path to $d$ with the maximal value. So for every node $v\in V$, let $(e_{1},\ldots,e_{s})$ be a path without cycles from $v$ to $d$ such that $e_{i}\in T$, and let $(e'_{1},\ldots, e'_{k})$ a different path from $v$ to $d$, then $\min_{1\leq i\leq s}W(e_{i})\geq \min_{1\leq j\leq k}W(e'_{j})$.
\end{itemize}

Notice that while for the minimal spanning tree and shortest path routing options weights are interpreted as costs (e.g. payments, delays), so these algorithms seek to minimize them, the weakest link views weights as measure for capability such as bandwidth, so seeks to maximize the weight.

\section{Axioms}
Having introduced our framework, we introduce our axioms, which are, basically, desirable properties of the function $f_{d}$ (in the axioms below we use $f$, as these are properties which do not depend on a specific $d$ destination).

\emph{Robustness} indicates the routing being quite unsusceptible to changes -- only if a path in the routing is destroyed, will it require any change. As indicated in Figure~\ref{robustEx}, the path from node $a$ changes, but not from node $b$.
\begin{axiom}[Robustness]
$f$ is \emph{robust} if removing an edge $e \in E$ from $G(V,E,W)$, yielding $G'$, then for every vertex $v\in V$: if the cycle-less path from $v$ to $d$ in $f_{d}(G)$ did not contain $e$, then this is still the selected path according to $f_{d}(G')$ (see example Figure~\ref{robustEx}).
\end{axiom}

\begin{figure}
\begin{center}
\includegraphics[scale=0.6]{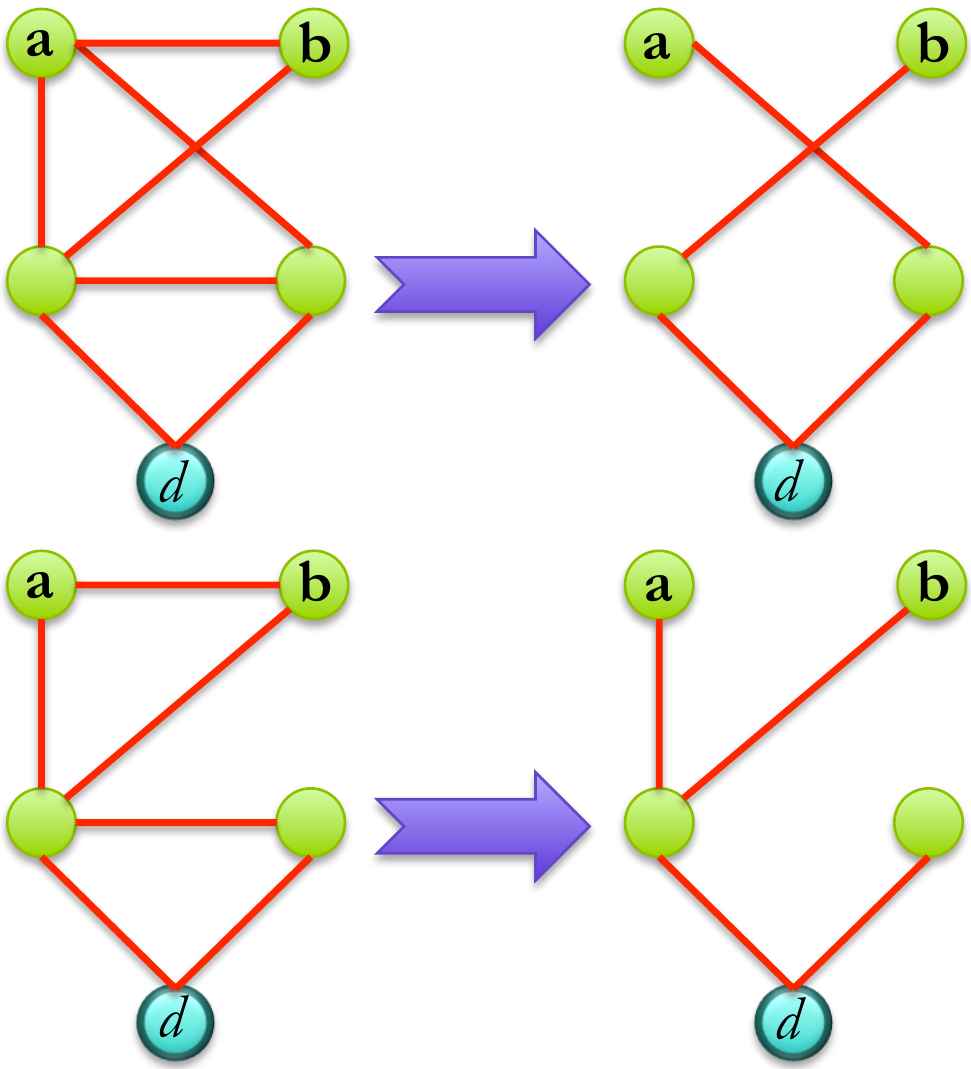}
 \end{center}
\caption {An edge is removed, but only $a$, whose path used that edge changes its path (the left side is the graph, the right side is the routing algorithm's output)}\label{robustEx}
\end{figure}

%

The following axioms deal with global changes to the graph weights, additive or multiplicative:

\begin{axiom}[Scale Invariance]
$f$ is \emph{scale invariant} if for a graph $G(V,E,W)$, for any positive scalar $\alpha\in\mathbb{R}_{+}$, defining $G'(V,E,\alpha W)$, for every $d\in V$, $f_{d}(G)=f_{d}(G')$.
\end{axiom}

\begin{axiom}[Shift Invariance]
$f$ is \emph{shift-invariant} if for a graph $G(V,E,W)$, for any $\alpha\in\mathbb{R}$, defining $G'(V,E,\alpha +W)$, for every $d\in V$, $f_{d}(G)=f_{d}(G')$.
\end{axiom}

The monotonicity axiom below seeks to establish that if an edge does not have to be in every tree, if its weight increases enough, it will not be a part of the routing tree:

\begin{axiom}[Monotonicity]
$f$ is \emph{monotone} if for a graph $G(V,E,W)$ and $d\in V$, for $e'\in E$, if $e'\notin f_{d}(G)$, then for every $G'(V,E,W')$, there is a value $M_{W'}$ such that for $W''$ such that $W''(e)=W'(e)$ for all $e\in E\setminus\{e'\}$ and $W''(e')\geq M_{W'}$, $e'\notin f_{d}(G''(V,E,W''))$.
Similarly, we can define the opposite direction, an edge in $f_{d}(G)$ will not be in the routing tree if it has a small enough value; we will refer to it as \emph{inverse monotonicity}.
\end{axiom}
While the phrasing of the following axiom is somewhat technical, the \emph{first hop} axiom below simply means that if a vertex has several potential edges to connect to a path to $d$, the routing only depends on the weights of the edges connecting it to these potential paths, and unrelated to weights of other edges in the graph.

\begin{axiom}[First Hop]
Let $G(V,E,W)$ be a weighted graph and let $v,d\in V$ and $d\neq v$. Suppose $C=\{c_{1},\ldots,c_{s}\}$ are the vertices such that $(v,c_{i})\in E$ and there is a path from $c_{i}$ to $d$ in $f_{d}(G)$ which does not pass through $v$. W.l.o.g., let $(v,c_{1})$ be the first step in the path from $v$ to $d$ in $f_{d}(G)$. We say that $f$ satisfies \emph{first hop} if for any $W'$ such that $W'(v,c_{i})=W(v,c_{i})$ and if for all $c_{i}\in C$ $f_{d}(G'(V,E,W'))$ contains paths to $d$ from $c_{i}$ that do not pass through $v$, and there is no $c'\notin C$ such that $(v,c')\in E$ and there is a path from $c'$ to $d$ in $f_{d}(G')$, then the cycle-less path from $v$ to $d$ in $f_{d}(G'(V,E,W'))$ starts with $(v,c_{1})$.
\end{axiom}

The rational for the \emph{first hop} axiom is to capture a common economic model, in which edge weights indicate the cost of passing information. In distributed networks, such as the internet, each agent only minds the amount it needs to pay to transfer its data to the next node, not caring about the path the data will take from there.

\emph{Path cardinal/ordinal invariance} intends to see the planner's considerations when multiple paths exist. As there might be many potential behaviours, we only limit ourselves to examining the narrow case of what the planner considers important when there is only one cycle in the graph (i.e., the axiom does not strongly enforce a general behaviour on the planner). Cardinal invariance deals with adding the same weight to potential paths, and how it does not effect the routing. Ordinal invariance similarly does not change the routing if all that has changed are the weights of the competing paths, as long as edges in each path maintain their relative position.

\begin{axiom}[Path Cardinal Invariance]
Let $G(V,E,W)$ be a graph which contains a single cycle, $d\in V$, and let $d\neq v\in V$ be a part of this cycle. Hence there are two alternative paths from $v$ to $d$ -- $p_{1}\subset E$ and $p_{2}\subset E$ (one of them is actually a part of $f_{d}(G)$). $f$ is \emph{path cardinal invariant} if it treats those paths as such: Choosing an edge $e'\in p_{1}$ and $e''\in p_{2}$, for any $\alpha\in \mathbb{R}$, we define $W'$ as $W(e)=W'(e)$ for $e\in E\setminus\{e' \cup e''\}$ and $W'(e')=W(e')+\alpha$ and $W'(e'')=W(e'')+\alpha$, the path from $v$ to $d$ will not change in $f_{d}(G(V,E,W'))$ (see example Figure~\ref{cardInv}).\\
\end{axiom}
\begin{figure}
\begin{center}
\includegraphics[scale=0.6]{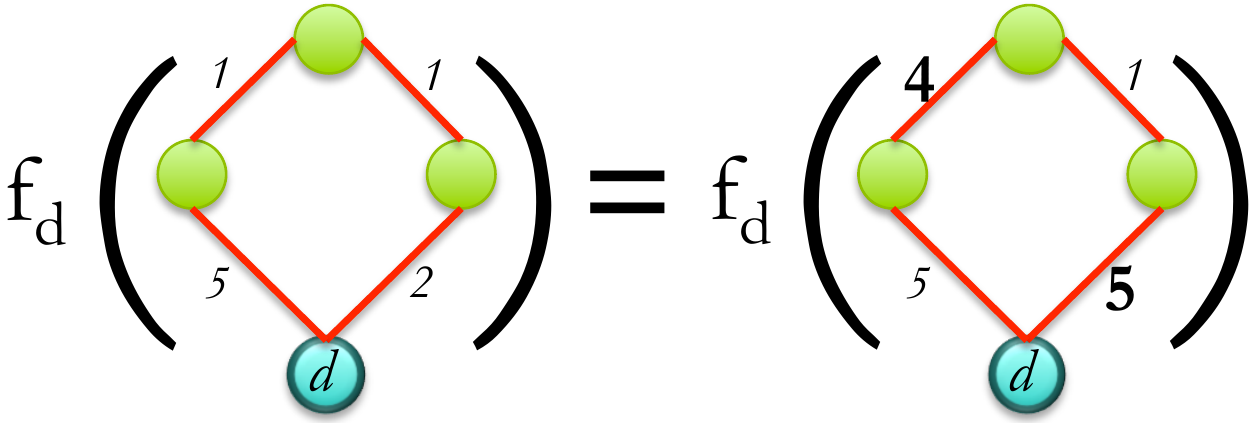}
 \end{center}
\caption {Selected path does not change when each path from the top node is added 2.}\label{cardInv}
\end{figure}

\begin{axiom}[Path Ordinal Invariance]
Let $G(V,E,W)$ be a graph which contains a single cycle, $d\in V$, and let $d\neq v\in V$ be a part of this cycle. Hence there are two alternative paths from $v$ to $d$ -- $p_{1}\subset E$ and $p_{2}\subset E$ (one of them is actually a part of $f_{d}(G)$). $f$ is \emph{path ordinal invariant} if it treats those paths as such: Taking an edge $e'\in p_{i}$ ($i\in \{1,2\}$) that is not maximal or minimal in $p_{1}\cup p_{2}$, we define $W'$ as $W(e)=W'(e)$ for $e\in E\setminus \{e'\}$ and allow W'(e') to be any value it chooses as long as for every $e''\in p_{i}$ if $W(e')\geq W(e'')$ then $W'(e')\geq W'(e'')=W(e'')$, and the path from $v$ to $d$ will not change in $f_{d}(G(V,E,W'))$ (see example Figure~\ref{ordInv}).\\
\end{axiom}
\begin{figure}
\begin{center}
\includegraphics[scale=0.6]{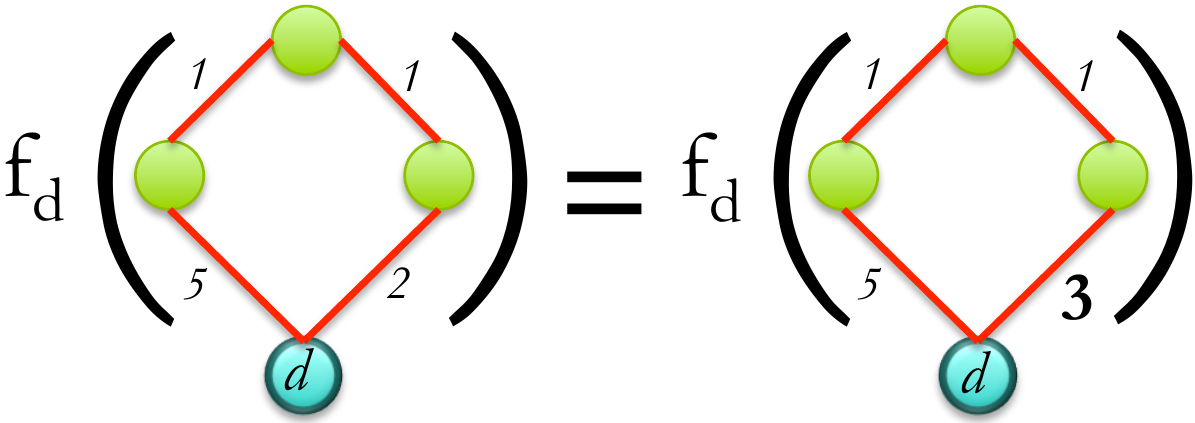}
 \end{center}
\caption {Selected path does not change when the bottom right edge is slightly increased.}\label{ordInv}
\end{figure}
\section{Minimal Spanning Tree}

\begin{theorem}\label{thm:MST}
A robust, scale invariant, shift invariant, monotone, first-hop (axioms 1-5) routing function $f$, for any graph $G(V,E,W)$ and $d\in V$, $f_{d}(G)$ will always be a minimal spanning tree of $G$.
\end{theorem}

\begin{reminder}
As our minimal spanning tree proof relies on the Kruskal algorithm, we will briefly describe it:
\begin{enumerate}
\item Order edges according to weights
\item Define a set $S$, initialized to the empty set.
\item Going over edges from lightest to heaviest, if the set $S\cup\{e\}$ has no cycles, $S=S\cup \{e\}$.
\end{enumerate}
\end{reminder}

\begin{proof}[Proof of Theorem~\ref{thm:MST}]
We shall prove the theorem using complete induction on the number of non-cycle lightest edges in the tree $f_{d}(G)$. Hence, we shall begin by proving that the lightest edge in the graph $G$ is in the routing tree $T=f_{d}(G)$. Assuming we are mistaken, let us consider the lightest edge in $G$ -- $e=(u,v)\in E$ -- and assume $e\notin T$. We create $G'(V,E',W)=f_{d}(G)\cup \{e\}$, and thanks to the robustness axiom, we know $f_{d}(G)=f_{d}(G')$.

If $v$'s path to $d$ in $f_{d}(G')$ goes through $u$, we shall switch the nodes' names, so that $v$'s path to $d$ does not pass through $u$. As $e$ is not in $f_{d}(G')$, there is an edge $e'=(u,s)$ that is the first step from $u$ towards $d$. We now define $x=W(e)$ and $y=W(e')$, and due to our minimality assumption, we know $x<y$.

Using the monotonicity axiom, we change graph weights to $W'$ that is identical to $W$ except that $e'$ weight is large enough so that we create a tree $T'$ in which there is a path from $s$ to $d$ that does not pass through $u$ (e.g., the same path that is in $f_{d}(G')$), and $v$ passes through $u$ towards $d$ (i.e., $e\in T'$).
 We define $y'=W''(e')$.

Using scale invariance we now multiply all edges by $\frac{y-x}{y'-x}$, and using shift invariance, we add to all edges $y-\frac{y-x}{y'-x}y'$. This means the weight of edge $e$ is now
$$
x\frac{y-x}{y'-x}+y-\frac{y-x}{y'-x}y'=(x-y')\frac{y-x}{y'-x'}+y=x
$$
While the weight of edge $e'$ is now
$$
y'\frac{y-x}{y'-x}+y-\frac{y-x}{y'-x}y'=y
$$

However, the routing tree contains $e$ and not $e'$, and a path from both $v$ and $s$ to $d$, contradicting the ``first hop'' axiom, which should have caused $e'$ to be chosen over $e$, as the edge weights for $e$ and $e'$ have not changed.

We now turn to the induction step -- we assume all bottom weighted $k-1$ edges that do not create a cycle are included in the tree $T=f_{d}(G)$, and we now seek to include the $k$-lightest edge that does not create a cycle. We pursue a similar path as we did as previously, and we shall mark the edge as $e=(u,v)$, and assume it is not included in $T=f_{d}(G)$, and instead $e'=(u,s)$ is included, and there is a path to $d$ from $v$ and $s$. Again, we create $G'(V,E',W)=f_{d}(G)\cup \{e\}$, and thanks to the robustness axiom, we know $f_{d}(G)=f_{d}(G')$. Using monotonicity we create weights $W'$ that just increase $e'$ weight, so that $G''=(V,E',W')$ has $T'=f_{d}(G'')$ which include the same bottom $k$ which do not create cycles (from the induction hypothesis), and $u$ reaches $d$ via the edge $e$. Recall that we know the bottom $k-1$ edges will definitely be in $f_{d}(G'')$, and we wish to ensure that there will still be a path from $v$ to $d$ and from $s$ to $d$. The same arguments used in the initial step of the induction ensure that, as well as returning the weights of $e$ and $e'$ to their values in $G$, while routing $u$ through $e$ and not $e'$ in the routing tree, reaching a contradiction with our initial assumption due to the ``first hop'' axiom.

What is left is to show MST indeed follows our axioms:
\begin{description}
\item[Robustness (axiom 1)] Trivial thanks to the Kruskal algorithm -- if the removed edge ($e'$) was not in the routing tree, it means it was not selected in the first place, and hence the same routing tree will be chosen. If it was, then any edge added after its removal ($e''$) closed a cycle with it, and hence, if affecting the edges in any path that did not include $e'$, it means $e''$ closes a cycle with them, hence $e'$ would have closed a cycle as well.
\item [Scale invariance (axiom 2)]  Multiplying all edges by a fixed amount does not change their order in relation to others, hence Kruskal will choose the same routing tree.
\item [Shift invariance (axiom 3)] Adding a fixed amount to all edges does not change their order in relation to others, hence Kruskal will choose the same routing tree.
\item [Monotonicity (axiom 4)] Giving an edge the maximal possible edge value ensures it will only be selected if no other edge can replace it -- and if there exists a tree without some edge, we know it will be chosen before.
\item [``First hop'' (axiom 5)] Kruskal ensures that if there are the same possible options of connecting a node to the tree, only the lightest edge will be chosen.

\end{description}
\end{proof}
\section{Shortest Path}

\begin{theorem}
A robust, scale invariant, monotone, and path cardinal invariant (axioms 1-2, 4, 6) routing function $f$, for any graph $G(V,E,W)$ and $d\in V$, $f_{d}(G)$ will always be a shortest path graph to $d$ of $G$.
\end{theorem}

\begin{proof}
Suppose $T=f_{d}(G)$ is not a shortest path routing tree. Let $u$ be the closest node to $d$ that is not connected to $d$ with a shortest path. Hence, there is an edge $e=(u,v)$ which will make $u$'s path a shortest path one ($v$, being closer to $d$, is already connected to $d$ with a shortest path), but $e\notin T$, and instead $e'=(u,s)$ is included in $T$.
Using robustness, we create $G'(V,E',W)=T\cup e$. $G'$ contains two alternate paths from $u$ to $d$, and $f_{d}(G)=f_{d}(G')$.

Using path cardinal invariant, we ``move'' all the value of the edges on each path to it's ``source'', i.e., to $(u,v)$ or $(u,s)$ (we do this by adding to the weight of $(u,v)$ and $(u,s)$ the value of\linebreak $\sum_{e\in (p_{1}\cup p_{2})\setminus(p_{1}\cap p_{2})}W(e)-W((u,v))-W((u,s))$, and reduce from $W((u,s))$ the weight of all edges of the path from $u$ to $d$ through $(u,v)$ and vice versa). We shall refer to $W(e)=x$ and $W(e')=y$. We now use monotonicity to create a new tree, with $e$ but without $e'$, with the graph's weight now $W'$ (identical to $W$ except for increase in $e'$ weight). Once again, we transfer all value of the paths from $u$ to $d$ to $e$ and $e'$ respectively, with everything else being 0. Now, using monotonicity, we increase the weight of $e'$ above that of $e$, with the weight of $(u,v)$ being $x$ (its path weights have not changed) and $(u,s)$ being $y'$.

Finally, we multiply all edges by $\frac{y-x}{y'-x}$ (using scale invariance), and using path cardinal invariance, we add to $e$ and $e'$ the amount $y-\frac{y-x}{y'-x}y'$. The weight of $e$ is now:
$$
x\frac{y-x}{y'-x}+y-\frac{y-x}{y'-x}y'=(x-y')\frac{y-x}{y'-x}+y=x
$$
While the weight of edge $e'$ is now
$$
y'\frac{y-x}{y'-x}+y-\frac{y-x}{y'-x}y'=y
$$

As all edges are the same weight as before, therefore we reached a contradiction regarding the inclusion of $e'$ instead of $e$ (whose weights are the same as well).

We will now show shortest path follows our axioms:
\begin{description}
\item[Robustness (axiom 1)] Removing an edge, at most, eliminates a potential path from a node to the destination $d$. If the path was not on the shortest path, the previous shortest path remains so.
\item [Scale invariance (axiom 2)] Multiplying by a fixed amount all edges means the value of each path is multiplied by the same amount, maintaining their relative ordering, hence what was shortest remains so.
\item [Monotonicity (axiom 4)] Giving an edge the value of the sum of all other edges ensures it will only be selected if no other path can replace it --- and if there exists a tree without some edge, we know there is such a path.
\item [Path cardinal invariance (axiom 6)] Having multiple paths from a node, adding the same amount to each path doesn't change the ordering of the paths (i.e., which path is ``shorter'' than another), hence selection of shortest path will be identical.
\end{description}
\end{proof}

\section{Weakest Link}

\begin{theorem}
A robust, scale invariant, shift invariant, inverse monotone, and path ordinal invariant (axioms 1-4, 7) routing function $f$, for any graph $G(V,E,W)$ and $d\in V$, $f_{d}(G)$ will always be a weakest link graph to $d$ of $G$.
\end{theorem}

\begin{proof}
Suppose $T=f_{d}(G)$ is not a weakest link routing tree. Let $u$ be a node that requires just one edge missing from $T$ that is not connected to $d$ with a weakest link\footnote{such a node exists as there is a node not connected by weakest link in $T$, hence adding the necessary path for that node, taking the node just before the final edge that we add to $T$ (i.e., closest to $d$), answers our criterion.}, and we mark this edge as $e=(u,v)$. Since $e\notin T$, there is an edge instead $e'=(u,s)$ that is included in $T$. Using robustness, we create $G'(V,E',W)=T\cup e$. $G'$ contains two alternate paths from $u$ to $d$, and $f_{d}(G)=f_{d}(G')$.

Using path ordinal invariant, we change the value of all edges on each alternate path from $u$ to $d$ to its weakest link value (we do this by taking the 2nd smallest edge in the path and changing its value to that of the weakest link, which by the axiom does not change the path chosen, and we proceed doing so to all edges on the path). We shall refer to $W(e)=x$ and $W(e')=y$ (from assuming $u$ is not in a weakest link path we know $x>y$). Using inverse monotonicity, we create $W'$ identical to $W$ except for $e'$ weight, that is low enough that it is not included in $f_{d}(G''(V,E',W'))$. Once again, we change the values of the paths from $u$ to $d$ to their weakest link value (this is only relevant for the path through $e'$, as the other path has not changed). We term the the new value for $e'$ -- $y'$, and we know $x>y'$.

Finally, we multiply all edges by $\frac{y-x}{y'-x}$ (using scale invariance), and using shift invariance, we add to all edges  $y-\frac{y-x}{y'-x}y'$. Edge $e$ now has the weight:
$$
x\frac{y-x}{y'-x}+y-\frac{y-x}{y'-x}y'=(x-y')\frac{y-x}{y'-x}+y=x
$$
While the weight of edge $e'$ is now
$$
y'\frac{y-x}{y'-x}+y-\frac{y-x}{y'-x}y'=y
$$

As all edges are the same weight as before, hence we reached a contradiction regarding the inclusion of $e'$ instead of $e$.

We shall now show weakest link also follows our axioms:
\begin{description}
\item[Robustness (axiom 1)] Removing an edge, at most, eliminates a potential path from a node to the destination $d$. If the path was not a weakest link, the previous weakest link remains so.
\item [Scale invariance (axiom 2)] Multiplying by a fixed amount all edges means the value of each path (its smallest edge) is multiplied by the same amount, maintaining their relative ordering, hence what was weakest link remains so.
\item [Shift invariance (axiom 3)] Adding a fixed amount all edges means the value of each path (its smallest edge) is added the same amount, maintaining their relative ordering, hence what was weakest link remains so.
\item [Monotonicity (axiom 4)] Giving an edge the value of the minimum of all other edges ensures it will only be selected if no other path can replace it --- and if there exists a tree without some edge, we know there is such a path.
\item [Path ordinal invariance (axiom 7)] Having multiple paths from a node, the weakest link edge (the one with smallest value) of the selected path can't become lower than the weakest link of the non-selected path, hence weakest link choice does not change.
\end{description}
\end{proof}

\section{Tightness of Axioms}
We will now show that the above characterizations are tight, and that without each axiom, other routing algorithms become possible.
  
\begin{theorem}
All MST axioms (1-5) are necessary, and without even one of them, other routing algorithms are possible.
\end{theorem}
\begin {proof}
Going over all MST axioms, we detail potential algorithms which work with all axioms except that one, and are not MST. We will refer below to each relaxed axiom, and to the new/additional system which can obtained by that relaxation:
\begin{description}
\item[Robustness] See example in Figure~\ref{robustMST}. Apply MST to any other graph that is not a linear transformation of the bottom one.
\item[Scale invariance] See example in Figure~\ref{multMST}. On all graphs except those which contain as a subgraph a linear transformations of the bottom one, apply MST.
\item[Shift invariance] See example in Figure~\ref{addMST}. On all graphs except those which contain as a subgraph a linear transformations of the bottom one, apply MST.
\item[Monotonicity] A maximal spanning tree implements all axioms but monotonicity.
\item[First Hop] Weakest link implements all of the other axioms.
\end{description}
\end{proof}

\begin{figure}
\begin{center}
\includegraphics[scale=0.6]{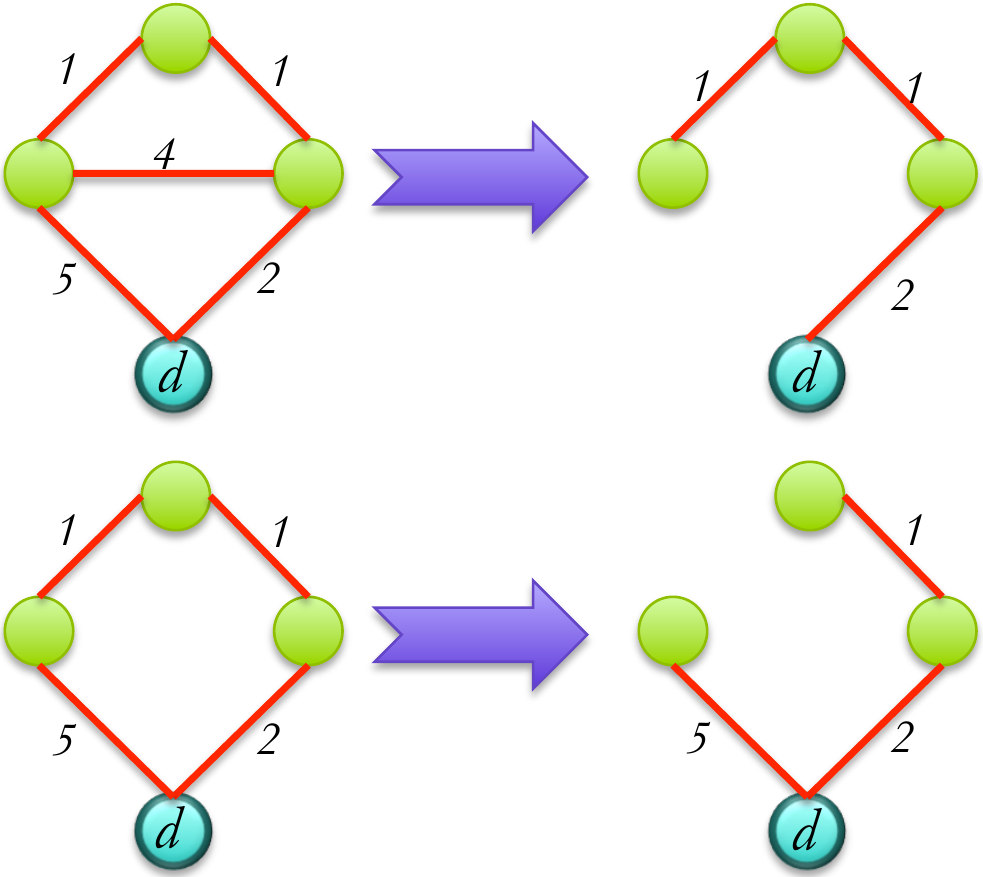}
 \end{center}
\caption {Lack of robustness results in a minimal spanning tree/shortest path routing (above) ending up in a routing tree that is weakest link but not MST or shortest path (below).}\label{robustMST}
\end{figure}
\begin{figure}
\begin{center}
\includegraphics[scale=0.6]{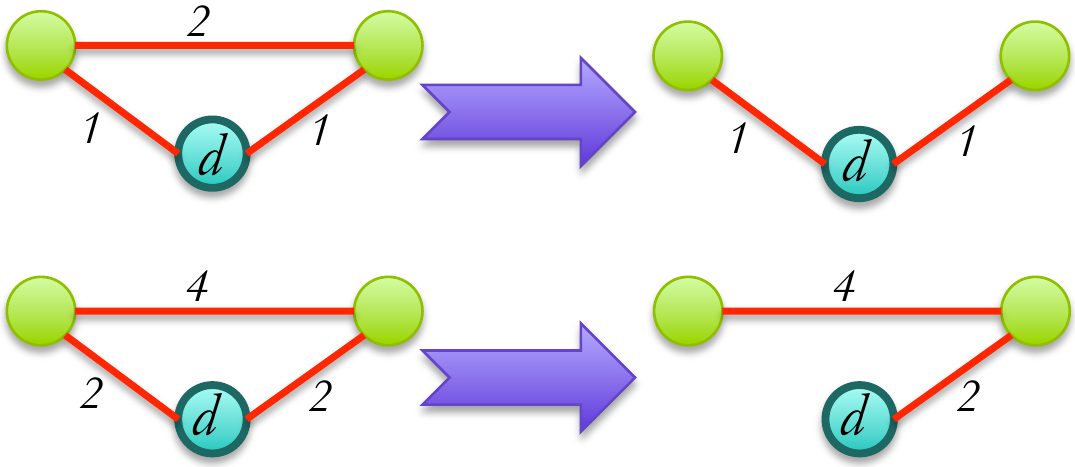}
 \end{center}
\caption {Eliminating scale invariance results in a minimal spanning tree/shortest path/weakest link routing (above) ending up in neither (below).}\label{multMST}
\end{figure}
\begin{figure}
\begin{center}
\includegraphics[scale=0.6]{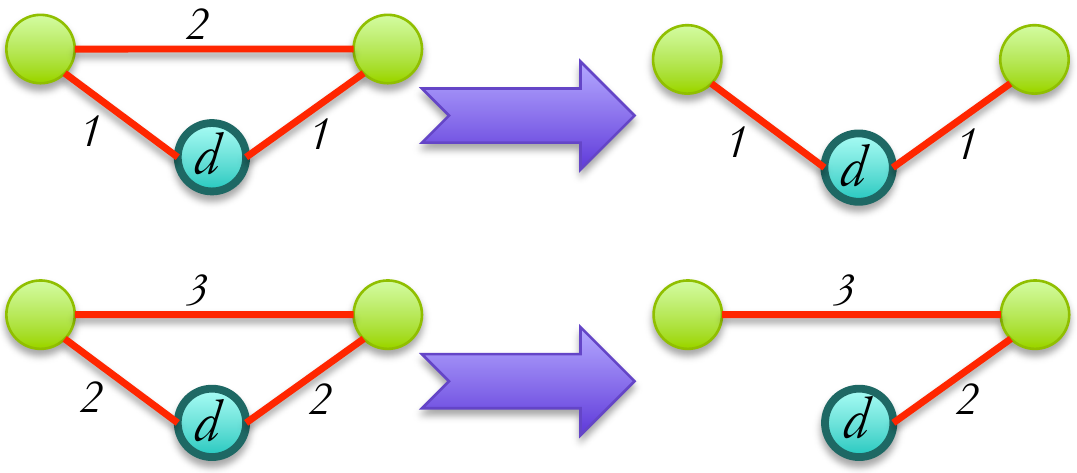}
 \end{center}
\caption {Eliminating shift invariance results in a minimal spanning tree/weakest link routing (above) ending up in neither (below).}\label{addMST}
\end{figure}

\begin{theorem}
All shortest path axioms (1-2, 4, 6) are necessary, and without even one of them, other routing algorithms are possible.
\end{theorem}
\begin {proof}
Going over all shortest path axioms, we detail potential algorithms which work with all axioms except one, and are not shortest path. We will refer below to the each relaxed axiom, and to the new/additional system which can obtained by that relaxation:
\begin{description}
\item[Robustness] See example in Figure~\ref{robustMST}. Apply shortest path to any other graph that isn't a scale of the structure of the bottom one. Any edge in that structure that is 100 times all the others is removed in the tree.
\item[Scale invariance] See example in Figure~\ref{multMST}. Taking the bottom example and for the group that includes all graphs for which it is a subgraph and those that can be formed by path cardinal invariance, and only for them do not apply shortest path but rather the example (it won't trample on the top example, as if the upper example adds $y$ to lower-right edge, and $y$ to the rest, and the bottom example adds $x$, it would require $2+x=4+x$, reaching an impossibility).
\item[Inverse Monotonicity] A longest path tree implements all axioms but monotonicity.
\item[Path cardinal invariant] Minimal spanning tree implements all other axioms.
\end{description}
\end{proof}

\begin{theorem}
All weakest link axioms (1-4, 7) are necessary, and without even one of them, other routing algorithms are possible.
\end{theorem}
\begin {proof}
Going over all weakest link axioms, we detail potential algorithms which work with all axioms except one, and are not weakest link. We will refer below to the each relaxed axiom, and to the new/additional system which can obtained by that relaxation:

\begin{description}
\item[Robustness] See example in Figure~\ref{robustMST}. Apply weakest link to any other graph that isn't of the structure of the bottom one, Any edge that in that structure that is 100 times less that all the others' weight is removed.
\item[Scale invariance] See example in Figure~\ref{multMST}. Taking the bottom example and for the group that includes all graphs for which it is a subgraph and those that can be formed by shift invariance and only for them do not apply weakest link but rather the example (it won't trample on the top example, as it can't be reached by shift invariance, and as the edge weights are all minimal/maximal, they change change by ordinal invariance).
\item[Shift invariance] See example in Figure~\ref{addMST}. Taking the bottom example and for the group that includes all graphs for which it is a subgraph and those that can be formed by scale invariance and only for them do not apply weakest link but rather the example (it won't trample on the top example, as it can't be reached by shift invariance, and as the edge weights are all minimal/maximal, they change change by ordinal invariance).
\item[Monotonicity] A strongest link tree implements all axioms but monotonicity.
\item[Path ordinal invariant] Minimal spanning tree implements all other axioms.
\end{description}
\end{proof}

\section{Discussion}
In this paper we explore the basic issue of routing -- how should information flow through a network and what properties might this process have. In the process of considering this issue we developed several properties we believe might be desirable by system planners. For example, \emph{robustness}, or the ability of a routing protocol to keep small changes from disrupting the whole routing process, is a property especially required in fast, changing networks.
 
Naturally, creating a structure from possible interactions between agents defined by a connections' graph is not limited just to information routing in networks such as the internet. Looking at organizations, where workers are connected according to their ability to work with other workers, and instead of routing messages between them we seek to construct an organizational hierarchy, we face a similar challenge. Again, robustness is a desirable property, as it means that if some workers have a worsening relationship with others, if they're not very senior in the organization, it has little effect on many others. In this case, we may consider the ``economic model'' (``first hop'' axiom) appropriate as well -- if workers only interact with their boss, we only care about the edge from each worker to his/her boss, and each worker does not care what happens further up in the hierarchy\footnote{Similarly, in a highly centralized organization, a path cardinal invariance is probably a sensible axiom.}.

Beyond setting up the axioms, we also examined common routing algorithms -- minimal spanning tree, shortest path and weakest link, and fully characterized them. Obviously, this is only the beginning of the road for this line of research -- further steps will entail developing more axioms and using them to characterize more algorithms, with the aim of giving a set of tools for system designers, allowing them to choose desirable properties which would dictate appropriate routing protocols.

\section*{Acknowledgments}
The authors thank Michael Schapira for his insightful discussions on this matter. Moshe Tennenholtz carried out this work while at Microsoft Research, Israel. Aviv Zohar is supported in part by the Israel Science Foundation (Grants 616/13 and 1773/13), and by the Israel Smart Grid (ISG) Consortium.

\bibliographystyle{eptcs}
\bibliography{general}
\end{document}